\documentclass[11pt]{article}
\usepackage{amsmath,amssymb,amsbsy,amsgen,amsopn,amsthm}
\usepackage{mathrsfs}
\usepackage[margin=1in]{geometry}

\usepackage{graphicx}
\usepackage{xspace}
\usepackage{graphicx}
\usepackage{algorithm,algorithmic}
\usepackage{subfig}
\usepackage{amssymb,amsmath}

\newcommand{\Rset}{\mathbb{R}}
\newcommand{\vcsp}{VC-weighted spanning tree problem\xspace}
\newcommand{\vcst}{VC-weighted Steiner tree problem\xspace}
\newcommand{\Afam}{\mathcal{A}}

\newtheorem{theorem}{Theorem}
\newtheorem{lemma}{Lemma}
\newtheorem{corollary}{Corollary}
\newtheorem{definition}{Definition}

\title{Computing a tree having a small vertex cover\footnote{A
preliminary version of this paper appeared in the proceedings of the 10th Annual International Conference on Combinatorial Optimization and Applications (COCOA'16).}}
\author{Takuro Fukunaga\thanks{RIKEN Center for Advanced Intelligence Project,
1-4-1 Nihonbashi, Chuo-ku, Tokyo 103-0022, Japan.
Email:~{\tt takuro.fukunaga@riken.jp}} \and Takanori Maehara\thanks{RIKEN Center for Advanced Intelligence Project,
1-4-1 Nihonbashi, Chuo-ku, Tokyo 103-0022, Japan.
Email:~{\tt takanori.maehara@riken.jp}}}

\date{}

\begin{document}

\maketitle

\begin{abstract}
 We consider a new Steiner tree problem, called vertex-cover-weighted Steiner tree problem.
 This problem defines the weight of a Steiner tree as the minimum weight
 of vertex covers in the tree, and seeks a minimum-weight Steiner
 tree in a given vertex-weighted undirected graph.
 Since it is included by the Steiner tree activation problem, the problem
 admits an $O(\log n)$-approximation algorithm in general graphs with
 $n$ vertices. This approximation factor is tight up to a constant because it is 
 NP-hard to achieve
 an $o(\log n)$-approximation for the vertex-cover-weighted Steiner tree
 problem in general graphs
 even if the given vertex weights are uniform and a spanning tree is required instead of a Steiner tree.
 In this paper, we present constant-factor approximation algorithms
 for the problem in unit disk graphs and in graphs excluding a fixed
 minor. For the latter graph class, our algorithm can be also applied
 for the Steiner tree activation problem.
 
 keywords: Steiner tree; unit disk graph; minor-free graph
\end{abstract}

\section{Introduction}
\label{sec:intro}

The problem of finding a minimum-weight tree in a graph 
has been extensively studied in the field of combinatorial optimization.
A typical example is the Steiner tree problem in edge-weighted graphs;
it has a long history of approximation algorithms,
culminating in the currently best approximation factor of
$1.39$~\cite{ByrkaGRS13,GoemansORZ12}.
The Steiner tree problem has also been studied in vertex-weighted graphs,
where
the weight of a Steiner tree is defined as the total weight of the
vertices spanned by the tree.
We call this problem the \emph{vertex-weighted Steiner tree problem}
while the problem in the edge-weighted graphs is called the
\emph{edge-weighted Steiner tree problem}.
There is an $O(\log k)$-approximation algorithm for the vertex-weighted Steiner tree problem
with $k$ terminals, and 
it is NP-hard to improve this factor because the problem includes the
set cover problem~\cite{DinurS14,KleinR93a}.

In this paper, we present a new variation of the Steiner tree problem.
Our problem is motivated by the following situation in communication
networks.
We assume that messages are exchanged
along a tree in a network;
this is the case in many popular routing protocols such as the spanning
tree protocol~\cite{Perlman85}.
We consider locating devices
that will monitor the traffic in the tree.
If a device is located at a vertex, it
can monitor all the traffic that passes through links incident to that
vertex.
How many devices do we need for monitoring all of the traffic in the tree?
Obviously, it depends on the topology of the tree.
If the tree is a star, it suffices to locate one device at the
center.
If the tree is a path on $n$ vertices,
then it requires $\lfloor n/2 \rfloor$ devices, because 
any vertex
cover of the path consists of at least $\lfloor n/2 \rfloor$ vertices.
Our problem is to compute a tree that minimizes the number (or, more generally, the weight) of devices required to monitor all of the traffic.

More formally, our problem is defined as
follows.
Let $G=(V,E)$ be an undirected graph associated with nonnegative vertex
weights
$w\in \Rset_+^V$.
Throughout this paper, we will denote $|V|$ by $n$.
Let $T \subseteq V$ be a set of vertices called \emph{terminals}.
The problem seeks a pair comprising a tree $F$ and a vertex set $U \subseteq
V(F)$ such that (i) $F$ is a Steiner tree with regard to the terminal
set $T$ (i.e., $T \subseteq V(F)$),
and (ii) $U$ is a vertex cover of $F$ (i.e., each edge in $F$ is
incident to at least one vertex in $U$).
The objective is to find such a pair $(F,U)$ that minimizes
the weight $w(U):=\sum_{v \in U}w(v)$ of the vertex cover.
We call this the \emph{vertex-cover-weighted (VC-weighted) Steiner tree
problem}.
We call the special case in which $V=T$ the
\emph{vertex-cover-weighted (VC-weighted) spanning tree problem}.
The aim of this paper is to investigate these fundamental problems.

Besides the motivation from the communication networks,
there is another reason for the importance of the
VC-weighted Steiner tree problem.
The VC-weighted Steiner tree problem is 
a special case of the \emph{Steiner tree activation problem}, which was
formulated by Panigrahi~\cite{Panigrahi11}.
In the Steiner tree activation problem,
we are given a set $W$ of nonnegative real numbers,
and each edge $uv$ in the graph is
associated with an activation function $f_{uv} \colon W \times W
\rightarrow \{\top,\bot\}$,
where $\top$ indicates that an edge $uv$ is activated, 
and $\bot$ indicates that it is not.
The activation function is assumed to be monotone
(i.e., if $f_{uv}(i,j)=\top$, $i \leq i'$, and $j \leq j'$,
 then $f_{uv}(i',j')=\top$).
A solution for the problem is defined as a $|V|$-dimensional 
vector $x \in W^V$.
We say that a solution $x$ \emph{activates} an edge $uv$ if $f_{uv}(x(u),x(v))=\top$.
The problem seeks a solution $x$ that minimizes $x(V):=\sum_{v\in
V}x(v)$
subject to the constraint that the edges activated by $x$ include a Steiner
tree.
 To see that the Steiner tree activation problem includes the
VC-weighted Steiner tree problem,
define
$W$ as $\{0\}\cup \{w(v) \colon v \in V\}$,
and let $f_{uv}(i,j)= \top$
if and only if $i \geq w(u)$ or $j\geq w(v)$ for each edge $uv$.
Under this setting, if $x$ is a minimal vector that activates an edge
set $F$, 
the objective $x(V)$ is equal to the minimum weight of vertex covers
of the subgraph induced by $F$.
Hence the Steiner tree activation problem under this setting
is equivalent to the VC-weighted  Steiner tree problem.

The Steiner tree activation problem models various natural settings in 
the design of wireless networks~\cite{Panigrahi11}.
Moreover, it includes several other well-studied problems.
One of them is the vertex-weighted
Steiner tree problem.
Indeed,
the vertex-weighted Steiner tree problem corresponds to the activation function
$f_{uv}$ such that
$f_{uv}(i,j)= \top$
if and only if $i \geq w(u)$ and $j\geq w(v)$ for each edge $uv$
the end vertices of which 
are associated with
vertex weights $w(u)$ and $w(v)$.
Note the similarity of the activation functions for the
VC-weighted and the vertex-weighted Steiner tree problems.
Thus the VC-weighted Steiner tree problem is an interesting variant of
the vertex-weighted Steiner tree and the Steiner tree activation
problems, which are studied actively in the literature.

In most of the known applications of the Steiner tree activation
problem (including the vertex-weighted and the VC-weighted
Steiner tree problems), $|W|$ is bounded by a polynomial of the input size.
Thus, it is usual to 
allow an algorithm for the algorithm to run in polynomial time in $|W|$.
Under this condition,
it is known that the Steiner tree
activation problems admits an $O(\log k)$-approximation algorithm when $|T|=k$.
Indeed,
Panigrahi~\cite{Panigrahi11}
gave an approximation-preserving reduction from the problem
to the vertex-weighted Steiner tree problem, and hence
the $O(\log k)$-approximation algorithm for the latter problem implies
that for the former problem.
This approximation factor is tight because
it is NP-hard to improve the factor for the vertex-weighted Steiner tree
problem, as mentioned above.
Even in the spanning tree variant of the Steiner tree activation problem,
the factor is proven to be tight~\cite{Panigrahi11}.

Since the VC-weighted Steiner tree problem is included by the Steiner tree activation problem,
the $O(\log k)$-approximation algorithm can also be applied to the
VC-weighted problem.
Moreover, Angel~et~al.~\cite{AngelBCK15} presented 
a
reduction from the dominating set problem
to the VC-weighted spanning tree problem with uniform vertex
weights.
This reduction implies that it is NP-hard to approximate the VC-weighted spanning tree problem
 within a factor of $o(\log n)$ even if the given vertex weights are uniform.
In Section~\ref{sec:general}, we present an alternative proof for this fact.
 
\subsection{Our contributions}

Because of the hardness of the VC-weighted spanning tree problem in general graphs, we will consider restricted graph classes.
We show that
the \vcst is NP-hard for unit disk graphs and planar graphs
(Theorem~\ref{thm.nphardness}).
Moreover, we present
constant-factor approximation algorithms
for the problem in
unit disk graphs
(Corollary~\ref{cor.unitdisk})
and
in graphs excluding a fixed minor (Theorem~\ref{thm.planer}).
Note that the latter graph class contains planar graphs.
For these graphs, it is known that
the vertex-weighted Steiner tree problem is NP-hard and admits constant-factor
approximation algorithms~\cite{DemaineHK09a,ZouLGW09,ZouLKW08}.
Hence it is natural to investigate approximation algorithms for the \vcst in these graph classes.
Moreover, unit disk graphs are regarded as a reasonable model
of wireless networks,
and the vertex-weighted Steiner tree problem in unit disk graphs
has been actively studied in this context (see, e.g., \cite{AmbuhlEMN06,HuangLS15,ZouLGW09,ZouLKW08,ZouWXLDWW11}).
Since our problem is motivated by an application in
communication networks,
it is reasonable to investigate the problem in unit disk graphs.

Our algorithm for unit disk graphs is based on a novel reduction to
another optimization problem.
The problem used in the reduction
is similar to the connected facility location problem studied in~\cite{EisenbrandGRS10,SwamyK04},
but it is slightly different.
In the connected facility location problem,
we are given sets $C,D \subseteq V$ of 
clients and facilities with an edge-weighted undirected graph $G=(V,E)$.
If a facility $f \in D$ is opened by paying an associated opening cost,
any client $i \in C$  can be allocated to $f$
by paying the allocation cost, which is defined as the shortest path length from
$i$ to $f$ multiplied by the demand of $i$.
The opened facilities must be spanned by a Steiner tree, which incurs a
connection cost defined as the edge weight of the tree multiplied by a
given multiplier $M$.
The objective is to find a set of opened facilities and
a Steiner tree connecting them, that minimizes the sum of the opening cost,
the allocation cost, and the connection cost.
Our problem differs from the connected facility location problem in the
fact that each client $i$ can be allocated to an opened facility $f$
only when $i$ is adjacent to $f$ in $G$, there is no cost for the
allocation,
and the multiplier $M$ for the connection cost is fixed to $1$.
It can be regarded as a combination of the dominating set and the
edge-weighted Steiner tree problems.
Hence we call this the \emph{connected dominating set problem},
although in the literature, this name is usually reserved for the case where
the connection cost is defined by vertex weights and all vertices in the
graph are clients.
From a geometric property of unit disk graphs, we show that
our reduction preserves the approximation guarantee up to a constant
factor if the graph is a
unit disk graph (Theorem~\ref{thm.reduction}).
To solve the connected dominating set problem,
we present a linear programming (LP) rounding algorithm.
This algorithm relies on an idea 
presented by Huang, Li, and Shi~\cite{HuangLS15},
who considered a variant of the connected dominating set problem in unit disk
graphs.
Although their algorithm is only for minimizing the number of vertices
in a solution,
we prove that it can be extended to our problem.

For graphs excluding a fixed minor, we solve the \vcst
by presenting a constant-factor approximation algorithm for 
the Steiner tree
activation problem.
Our algorithm simply combines
the reduction to the vertex-weighted Steiner tree problem~\cite{Panigrahi11}
and 
the algorithm of Demaine,
Hajiaghayi, and Klein~\cite{DemaineHK09a} for the vertex-weighted
Steiner tree problem in graphs excluding a fixed minor.
However, analyzing it is not straightforward,
because the reduction does not preserve the minor-freeness of the input graphs.
Nevertheless, we show that the algorithm of Demaine~et~al.\ achieves a
constant-factor approximation for the graphs constructed by the
reduction
(Section~\ref{sec.planer}).

\subsection{Organization}

The remainder of this paper is organized as follows.
Section~\ref{sec:prelim} introduces the notation and preliminary facts used
throughout the paper.
Section~\ref{sec:general}
presents hardness results on the \vcst.
Sections~\ref{sec.steiner-unitdisk} and~\ref{sec.planer}
provide constant-factor approximation
algorithms
for unit disk graphs and for graphs excluding a fixed minor, respectively.
Section~\ref{sec:conclusion} concludes the paper.

\section{Preliminaries}
\label{sec:prelim}

We first define the notation used in this paper.
Let $G=(V,E)$ be a graph with the vertex set $V$ and the edge set $E$.
We sometimes identify the graph $G$ with its edge set $E$ 
and by $V(G)$ denote the vertex set of $G$.
When $G$ is a tree, $L(G)$ denotes the set of leaves of $G$.

Let $U$ be a subset of $V$.
Then $G-U$ denotes the subgraph of $G$
obtained
by removing all vertices in $U$ and all edges incident to them.
$G[U]$ denotes the subgraph of $G$ induced by $U$.

We denote a singleton vertex set $\{v\}$ by $v$.
An edge joining vertices $u$ and $v$ is denoted by $uv$.
For a vertex $v$, $N_G(v)$ denotes the set of neighbors of $v$ in a
graph $G$, i.e.,
$N_G(v)=\{u \in V \colon uv \in E\}$.
$N_G[v]$ indicates $N_G(v) \cup v$.
We let $d_G(v)$ denote $|N_G(v)|$.
For a set $U$ of vertices,
$N_G(U)$ denotes $(\bigcup_{v \in U}N_G(v))\setminus U$.
When the graph $G$ is clear from the context,
we may remove the subscripts from our notation.
We say that a vertex set $U$ \emph{dominates} a vertex $v$
if $v \in U$, or $U$ contains a vertex $u$ that is adjacent to $v$.
If a vertex set $U$ dominates each vertex $v$ in another vertex set $W$,
then we say that $U$ dominates $W$.

A graph $G$
is a unit disk graph when
there is an embedding of the vertex set 
into the Euclidean plane
such that
two vertices $u$ and $v$ are joined by an edge
if and only if their Euclidean distance
is at most 1.
If $G$ is a unit disk graph, we call such an embedding
a \emph{geometric representation} of $G$.

Let $G$ and $H$ be undirected graphs.
We say that $H$ is a \emph{minor} of $G$ if $H$ is obtained from $G$
by deleting
edges and vertices and by contracting edges.
If $H$ is not a minor of $G$, $G$ is called \emph{$H$-minor-free}.
By Kuratowski's theorem, a graph is planar if and only if it is
$K_5$-minor-free and $K_{3,3}$-minor-free.

As mentioned in Section~\ref{sec:intro},
the Steiner tree activation problem 
contains both the VC-weighted and the vertex-weighted Steiner tree problems.
In addition,
Panigrahi~\cite{Panigrahi11}
showed that the Steiner tree activation problem can be reduced to the
vertex-weighted Steiner tree problem.
Since we use this reduction later,
we present it in the following theorem.

 \begin{theorem}[\cite{Panigrahi11}]
  \label{thm.reduction}
 There is an approximation-preserving reduction from the Steiner tree
 activation problem to the vertex-weighted Steiner tree problem.
 Hence, if the latter problem admits an $\alpha$-approximation algorithm,
 the former problem does also.
 \end{theorem}
 \begin{proof}
Recall that an instance $I$ of the Steiner tree activation problem
consists of 
an undirected graph $G=(V,E)$, a terminal set $T$,
a range $W \subseteq \Rset_+$,
and an activation function $f_{uv}\colon W \times W\rightarrow \{\top,\bot\}$
for each $uv \in E$.
We define a copy $v_i$
of a vertex $v$ for each $v \in V$ and $i \in W$,
and associate $v_i$ with the weight $w(v_i):=i$.
We join $u_i$ and $v_j$ by an edge
if $uv \in E$ and $f_{uv}(i,j)=\top$.
In addition, we join each terminal $t \in T$
with its copies $t_i$, $i \in W$.
The weight $w(t)$ of $t$ is defined to be $0$.
Let $G'$ be the obtained graph on the vertex set
$T \cup \{v_i \colon v \in V, i \in W\}$.
Let $I'$ be the instance of
the vertex weighted Steiner tree problem
that consists of the graph $G'$, the vertex weights $w$, and the
terminal set $T$.
From an inclusion-wise minimal Steiner tree $F$ feasible to $I'$,
define a vector $x \in W^V$
by $x(v)=\max\{i \in W \colon v_i \in V(F)\}$ for each $v \in V$.
Then $x$ activates a Steiner tree in the original instance
$I$, and $x(V)$ is equal to the vertex
weight of $F$.
Hence
there is a one-to-one correspondence between
a minimal Steiner tree in $I'$
and a feasible solution in $I$,
and they have the same objective values in their own problems.
Hence the above reduction is an approximation-preserving reduction from the Steiner
tree activation problem to the vertex-weighted Steiner tree problem.
\end{proof}

We note that the reduction
claimed in Theorem~\ref{thm.reduction}
transforms the input graph,
and hence it may not be closed in a graph class.
In fact, we can observe that the reduction is not closed in
unit disk graphs or planar graphs.

\section{Hardness of VC-weighted spanning tree and Steiner tree problems}
\label{sec:general}

In this section, we present hardness
results of the VC-weighted spanning tree and Steiner tree problems.
First, we prove that it is NP-hard to approximate
the \vcsp
within a factor of $o(\log n)$ even if vertex weights are uniform.
This fact has already been proven by Angel~et~al.~\cite{AngelBCK15}.
Here, we give an alternative proof which consists of an approximation-preserving reduction from the set cover
problem.

 \begin{theorem}
  \label{thm:hardness}
  There exists a constant $c$ such that
 it is NP-hard to approximate \vcsp
 within a factor of $c\log n$ even if the given vertex weights are uniform.
 \end{theorem}
\begin{proof}
 Recall that an instance of the set cover problem consists of
a finite set $S$ and a family $\mathcal{X}$ of subsets of $S$.
The objective of the problem
 is to find a subfamily $\mathcal{X}'$ 
 of $\mathcal{X}$
 such that $\bigcup_{X \in \mathcal{X}'}X = S$ and
 $|\mathcal{X}'|$ is minimized.
 A feasible solution for the set cover instance is called a set cover.

For each set $X \in \mathcal{X}$,
we define a vertex $v_X$ corresponding to it.
Let $W$ denote $\{v_X \colon X \in \mathcal{X}\}$.
The set cover
instance defines a bipartite graph 
on the vertex set $S \cup W$;
two vertices $v \in S$ and $v_X \in W$ are joined by an edge
if and only if $v \in X$.
To this bipartite graph,
we add edges so that 
any two vertices in $W$ are adjacent.
Let $G=(S \cup W, E)$ be the obtained graph.
 We reduce the set cover instance to the instance of
\vcsp
 on this graph with uniform vertex weights.
 We prove that this reduction is approximation-preserving.
 For every $\epsilon > 0$,
 it is NP-hard to approximate the set cover instance within a
 factor of $(1-\epsilon)\ln |S|$~\cite{DinurS14}, 
 and this hardness holds for instances such that $|\mathcal{X}|$ is polynomial 
 on $|S|$.
 Thus, the reduction proves the theorem.

Let $\mathcal{X}'$ be a set cover for instance $(S, \mathcal{X})$.
We define $W':=\{v_X \in W \colon X \in \mathcal{X}'\}$ from $\mathcal{X}'$.
Since $\mathcal{X}'$ is a set cover, 
for each $u \in S$, 
there exists $v_X \in W'$ that is adjacent to $u$ in $G$.
We define $F$ as the set of
 edges joining such pairs of $u \in S$ and $v_X \in W'$.
 Let $F'$ be a star on $W$ such that its center is an arbitrary vertex
 in $W'$ and $F'$ spans all vertices in $W$.
 Then $F\cup F'$ is a spanning tree in $G$,
 and $W'$ is a vertex cover in $F \cup F'$.
 Thus, for any set cover $\mathcal{X}'$,
 there exists a feasible solution $(F\cup F',W')$
 with $|W'|=|\mathcal{X}'|$
 for the instance of \vcsp.
 
Let us consider the other direction.
Let $(F,U)$ be a solution for the instance of \vcsp.
If $U \subseteq W$, then
$\mathcal{X}':=\{X \in \mathcal{X} \colon v_X \in U\}$
is a set cover in $(S,\mathcal{X})$
because each vertex $v \in S$ is adjacent to a vertex $v_X \in U$ 
in $F$, and hence $\mathcal{X}'$ contains a set $X$ with
$v \in X$.
Notice that $|U|=|\mathcal{X}'|$.

Suppose that $U$ contains a vertex $u \in S$. Let $v_{X_1},\ldots,v_{X_k}$ be the
vertices in $W$ that are adjacent to $u$ on $F$.
We prove that the solution $(F,U)$ can be modified to another feasible
solution $(F',U')$ with $|U' \cap S| < |U \cap S|$ and $|U'|\leq |U|$.
By repeating this modification, we obtain a feasible solution whose
vertex cover is contained by $W$.
We define $U'$ as $(U \setminus u) \cup v_{X_1}$.
Let $F'$ be the edge set obtained from $F$ by replacing all edges
$uv_{X_2},\ldots,uv_{X_k}$
with $v_{X_1} v_{X_2},\ldots,v_{X_1} v_{X_k}$.
Then $F'$ is a spanning tree, and
$U'$ is a vertex cover on $F'$.
\end{proof}

 As noted in Theorem~\ref{thm.reduction},
 there is an approximation-preserving reduction from the Steiner tree
 activation problem to the vertex-weighted Steiner tree problem,
 and the latter problem admits an $O(\log |T|)$-approximation algorithm
 in general graphs.
 Since the Steiner tree activation problem includes
 \vcst, this indicates that \vcst
 also admits an $O(\log |T|)$-approximation algorithm.
By Theorem~\ref{thm:hardness}, the approximation factor achieved by this
algorithm is tight up to a constant.

Next, we consider unit disk graphs and planar graphs.
We show that the \vcst is NP-hard for these graph classes.

 \begin{theorem}
  \label{thm.nphardness}
 \vcst is NP-hard for unit disk graphs and for planar graphs.
 \end{theorem}
\begin{proof}
 Garey and  Johnson~\cite{GareyJ77} proved that the edge-weighted Steiner tree problem
 is NP-hard even in the grid graphs.
 We show that the edge-weighted Steiner tree problem in grid graphs
 can be reduced to the \vcst in unit disk graphs.
 We can suppose without loss of generality that the distance between every two adjacent vertices $u$ and
 $v$ in the  grid graph
 is 4. 
 For each pair of adjacent vertices $u$ and $v$,
 we
 subdivide the edge $uv$
 by adding three new vertices $i,j$, and $k$ distributed equally between
 $u$ and $v$ as illustrated in Figure~\ref{fig.reduction}.
 The graph obtained by this way is a unit disk graph
 because two vertices in the graph are adjacent if and only if the
 distance between them is exactly a unit length.
 From the edge weights $w'$ of the original graph,
 we define the vertex weights $w$ of the new graph by
 $w(u)=w(v)=0$, $w(i)=w(k)=+\infty$, and $w(j)=w'(ij)$.
 Then, if edges $ui$, $ij$, $jk$, and $kv$ are included in a Steiner
 tree,
 a minimum-weight vertex cover on the tree includes $u$, $v$, and $j$.
 Hence, the minimum weight of vertex covers on a Steiner tree in the unit
 disk graph is equal to the edge weight of the corresponding Steiner
 tree in the original graph.
 Hence this gives an approximation-preserving reduction from the edge-weighted Steiner
 tree problem in grid graphs to the \vcst in unit disk graphs.

 Notice that the graph constructed by the above reduction is also planar.
 Hence this also proves the NP-hardness of \vcst in planar graphs.
\end{proof}

 \begin{figure}
  \centering
 \includegraphics[scale=.65]{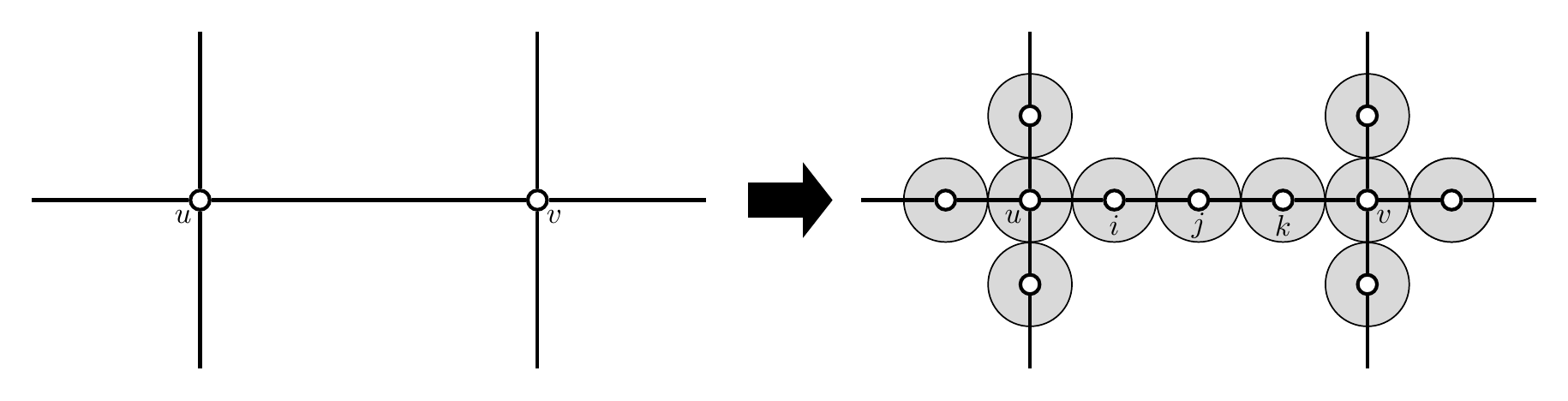}
 \caption{Reduction from the edge-weighted Steiner tree problem in grid
 graphs to the \vcst in planar unit disk graphs}
 \label{fig.reduction}
 \end{figure}

\section{\vcst in unit disk graphs}
\label{sec.steiner-unitdisk}

The aim of this section
is to
present a constant-factor
approximation algorithm for the \vcst in unit disk graphs.
Our algorithm consists of two steps. In the first step, we reduce the \vcst 
to another optimization problem, which is called the connected dominating set 
problem.
We present this reduction in Section~\ref{sec:reduction}.
When the original problem is the \vcsp,
the connected dominating set problem can be solved by a simpler
algorithm, which we will explain in Section~\ref{sec.spanning}.
Then, in Section~\ref{sec.lp-rounding}, we will present an LP-rounding algorithm for the general case of the
connected dominating set problem.

  \subsection{Reduction}
\label{sec:reduction}

As noted in Theorem~\ref{thm.reduction}, the Steiner tree activation
problem can be reduced to the vertex-weighted Steiner tree problem.
Since the \vcst is included in the Steiner tree activation problem,
the reduction also applies to the \vcst. Since there is a constant-factor
approximation algorithm for the vertex-weighted Steiner
tree problem in unit disk graphs, this reduction gives a constant-factor
approximation for the VC-weighted problem if the graph constructed by the reduction is a
unit disk graph. However, the constructed graph may not be a unit disk
graph, even if the original graph is a unit disk graph.
This can be seen through an example. Let $G=(V,E)$ be a star graph 
with the center vertex $v$. We do not specify the terminal set $T$,
because
it is not important here.
When the original
 problem is the \vcst, the reduction given in Theorem~\ref{thm.reduction}
can be simplified as follows.
 Two copies $u_{\circ}$ and $u_{\bullet}$ are constructed from each vertex $u \in V(G)$,
 where $u_{\circ}$ indicates that $u$ is not included in
 the vertex cover of the solution,
 and $u_{\bullet}$ indicates that it is.
 For each edge $uu' \in E$, the graph constructed by the reduction
 contains edges $u_{\bullet}u'_{\bullet}$, 
 $u_{\circ}u'_{\bullet}$, and $u_{\bullet}u'_{\circ}$.
 Moreover, each terminal $t \in T$ is adjacent to its copies
 $t_{\circ}$ and $t_{\bullet}$.
 Let $G'$ be a graph on the vertex set $T \cup
 \{u_{\circ},u_{\bullet}\colon u \in V\}$ that is constructed in this way.
 By solving the vertex-weighted Steiner tree problem on $G'$,
 we can compute a solution to the VC-weighted problem on $G$. 
 If the degree of $v$ is at most 5, $G$ is a unit disk graph.
 The degree of $v_{\bullet}$ in $G'$ is twice the degree of $v$ in $G$,
 and any two neighbors of $v_{\bullet}$ are not adjacent in $G'$.
 Hence $G'$ contains $K_{1,6}$ as an induced subgraph if the
 degree of $v$ in $G$ is at least 3.
 Since no unit disk graph contains $K_{1,6}$ as an induced subgraph,
 this means that $G'$ is not a unit disk graph.

 Our idea is to reduce the \vcst to another optimization problem.
 This is inspired by a constant-factor approximation algorithm for the
 vertex-weighted Steiner tree problem on a unit disk graph~\cite{ZouLGW09,ZouLKW08},
 which is based on a reduction from the vertex-weighted to
 the edge-weighted Steiner tree problems.
 The reduction is possible because the former problem always admits an optimal
 Steiner tree in which the maximum degree is a constant if the graph is
 a unit disk graph.
 Even in the \vcst, if there is an optimal solution $(F,U)$
 such that the maximum degree of vertices in the vertex cover $U$ is a
 constant in the Steiner tree $F$, then 
 we can reduce the problem to the edge-weighted Steiner tree problem.
 However,
 there is an instance of the \vcst
 that admits no such optimal solution.
 For example, if the vertex weights are uniform, and the graph includes a
 star in which all of the terminals are its leaves,
 then the star is the Steiner tree in the optimal solution,
 and its minimum vertex cover consists of only the center of the star.
 The degree of the center of the star is not bounded by a constant.
 Hence it seems that it would be difficult to reduce the \vcst to the edge-weighted problem.

 We reduce the \vcst to a problem similar to the connected facility location problem.
 The reduction is based on a geometric
 property of unit disk graphs, and we will begin by proving this property.
 The following lemma gives a basic claim about geometry.
 For two points $i$ and $j$ on the plane, we denote their Euclidean
 distance by $l_{ij}$.

\begin{lemma}\label{lem.geometry}
 Let $i$ be a point on the Euclidean plane,
 and let $\alpha \in (1/2,3/4]$.
 Let $P$ be a set of points on the plane such that $\alpha < l_{ik}/l_{ij}\leq
 1/\alpha$ holds for all $j,k \in P$.
 If $|P|  > 2\pi/\arccos (\alpha/2 + 3/(8\alpha))$,
 then there exist $j,k \in P$
 such that $l_{jk} < \max\{l_{ij}, l_{ik}\}/2$.
\end{lemma}
\begin{proof}
 Since $|P| > 2\pi/\arccos (\alpha/2 + 3/(8\alpha))$, there exist $j,k \in P$ such that $\theta := \angle
 jik < \arccos (\alpha/2 + 3/(8\alpha))$.
 We note that $l^2_{jk}=l^2_{ij} + l^2_{ik} -2 l_{ij} l_{ik} \cos
 \theta$.
 Without loss of generality, we assume $l_{ij} \geq l_{ik}$.
 Then, $(\max\{l_{ij},l_{ik}\})^2=l^2_{ij}$.
 Hence it suffices to show that
 $-4l^2_{ik}- 3l^2_{ij} + 8l_{ij}l_{ik} \cos \theta > 0$.

  Let $\beta:=l_{ik}/l_{ij}$. Then, $\alpha < \beta \leq 1$ holds.
$\sup_{\beta:\alpha < \beta \leq 1}4\beta +
 3/\beta = 4\alpha+3/\alpha$ holds.
 Hence the required inequality is verified by
 \begin{align*}
  -4l^2_{ik}- 3l^2_{ij} + 8l_{ij}l_{ik} \cos \theta
  & = l_{ij}l_{ik} \left(
  -4 \beta - \frac{3}{\beta} +8 \cos\theta
  \right)\\
  & \geq l_{ij}l_{ik} \left(
  -4 \alpha - \frac{3}{\alpha} +8 \cos\theta
  \right)\\
  & > 0.
 \end{align*}
\end{proof}

Our reduction requires the assumption that
there is an optimal solution $(F,U)$ for the \vcst
such that
the degree of each vertex $v \in U$ is bounded by a constant $\alpha$
in the tree $F-(L(F)\setminus U)$.
The following lemma proves that
the assumption holds with $\alpha=29$ if the input graph is a unit disk graph.

 \begin{lemma}\label{lem.constant-degree}
  If the input graph $G=(V,E)$ is a unit disk graph,
  the \vcst
  admits an optimal solution consisting of a Steiner tree $F$
  and a vertex cover $U$ of $F$ such that
  the degree of each vertex in $U$ is at most $29$ in
  $F - (L(F)\setminus U)$.
 \end{lemma}
\begin{proof}
 For two vertices $u,v\in V$, 
 let $l_{uv}$ denote the Euclidean distance between $u$ and $v$
 in the geometric representation of $G$.
 Let $(F,U)$ be an optimal solution for the \vcst.
 We call 
 each node in $V(F)\setminus (L(F)\cup U)$
 an \emph{inner node of $(F,U)$}.
 Without loss of generality, we can assume that $(F,U)$ satisfies the
 following conditions:
 \begin{itemize}
  \item[(a)] $(F,U)$ minimizes the number of inner nodes over all optimal
	     solutions;
  \item[(b)]
	    $F$ minimizes $\sum_{e \in F}l_e$ over all optimal solutions
	    subject to (a);
  \item[(c)] $(F,U)$ minimizes
	     the number of vertices $v \in U$
	     such that $|\{u \in U \colon uv \in F\}|\geq 6$ over all
	     optimal solutions subject to (a) and (b).
 \end{itemize}
 Let $v \in U$.
 Let $M_v:=\{u \in U \colon uv \in F\}$ and $M'_v:=\{u \in
 V(F)\setminus (U \cup L(F)) \colon uv \in F\}$.
 We prove the lemma by showing that $|M_v|\leq 5$ and $|M'_v|\leq 24$.

 We first show that $|M_v|\leq 5$. Suppose that there are two distinct vertices $i,j \in M_v$ such that 
 $l_{ij} < \max\{l_{vi}, l_{vj}\}$.
 Without loss of generality, let  $l_{vi}=\max\{l_{vi}, l_{vj}\}$,
 and denote $F \setminus \{vi\} \cup \{ij\}$ by $F'$.
 Then $F'$ is a Steiner tree and 
 $U$ is a vertex cover of $F'$.
 That is to say, $F'$ is another optimal solution for the problem.
 Moreover, $(F',U)$ has the same set of inner nodes as $(F,U)$, 
 and $\sum_{e \in F'}l_e < \sum_{e \in F}l_e$.
 Since the existence of such an optimal solution contradicts condition (b),
 $M_v$ contains no such vertices $i$ and $j$.

 If $|M_v| \geq 7$, there must be two vertices $i,j \in M_v$
 such that $\angle ivj < \pi/3$, and
 $l_{ij} < \max\{l_{vi}, l_{vj}\}$ holds 
 for these vertices. 
 Hence $|M_v| \leq 6$ holds.
 Suppose that $|M_v| =6$.
 In this case,
 $M_v=\{u_1,\ldots,u_6\}$, and 
 $l_{vu_k}=l_{vu_{k+1}}=l_{u_k u_{k+1}}$ holds for all
 $k=1,\ldots,6$, where, for notational convenience, we let $u_{7}$ denote $u_1$.
 If $|M_{u_1}|\leq 4$,
 we define $F'$ as $F \setminus \{vu_2\} \cup \{u_1u_2\}$.
 Then, $(F',U)$ is another optimal solution
 that has the same inner node set as $(F,U)$,
 and $\sum_{e \in F'}l_e = \sum_{e \in F}l_e$.
 Replacing $F$ by $F'$
 decreases the number of vertices $v \in U$
 such that $|M_v|\geq 6$, which contradicts condition (c).
 If $|M_{u_1}|\geq 5$, then (i) there exist $i,j \in M_{u_1}\setminus v$
 such that $l_{ij} < \max\{l_{u_1 i}, l_{u_1j}\}$,
 or (ii) there exist $i \in M_{u_1} \setminus v$
 and $j \in \{v,u_2,u_6\}$
 such that
 $l_{ij} < \max\{l_{u_1 i}, l_{u_1j}\}$.
 Case (i) contradicts condition (b), as observed above.
 In case (ii),
 we define $F'$ as $F \setminus \{u_1 i\} \cup \{ij\}$
 if $\max\{l_{u_1 i}, l_{u_1j}\}=l_{u_1 i}$,
and as $F \setminus \{u_1 v\} \cup \{ij\}$
 if $\max\{l_{u_1 i}, l_{u_1j}\}=l_{u_1 j}$.
 In either case, $(F',U)$ is another optimal solution
 that has the same inner node set as $(F,U)$,
 and $\sum_{e \in F}l_e >\sum_{e \in F'}l_e$. Since this contradicts 
 condition (b), $|M_v| \leq 5$ holds.
 
 \begin{figure}\centering
  \includegraphics[scale=.9]{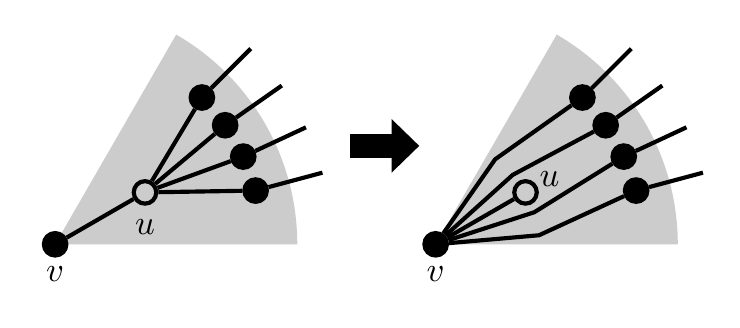}
  \caption{Transformation of $F$ when $l_{vu'}\leq 1$ for all $u' \in A_u$ and $u \in T$}
  \label{fig.proof1}
 \end{figure}

 Next, we prove $|M'_v|\leq 24$. Let $u \in M'_v$.
 Since $u$ is not a leaf, $u$ has a neighbor other
 than $v$. We denote by $A_u$ the set of neighbors of $u$ other than $v$.
 Since $u \not\in U$, each vertex in $A_u$ is included in $U$.
 If
 $l_{v u'}  \leq 1$ holds
 for all vertices $u' \in A_u$,
 consider
 $F'$ defined 
 as $(F \setminus \{uu' \colon u' \in A_u\}) \cup \{vu' \colon u'\in A_u\}$
 when $u$ is a terminal (see Figure~\ref{fig.proof1}),
 and 
 as $(F \setminus (\{vu\}\cup\{uu' \colon u' \in A_u\})) \cup \{vu' \colon u'\in A_u\}$
 otherwise.
 Then, $F'$ is a Steiner tree and $U$ is a vertex cover of $F'$,
 and hence $(F',U)$ is another optimal solution for the problem.
 Moreover,
 all inner nodes of $(F',U)$ are also inner nodes of $(F,U)$,
 and $u$ is not inner nodes of $(F',U)$ but $(F,U)$.
 This means that $(F',U)$ has fewer inner nodes than $(F,U)$.
 Since the existence of such an optimal solution contradicts condition (a),
 there is at least one vertex $u' \in A_u$ with $l_{vu'} > 1$.
 We choose one of these vertices for each $u \in M'_v$,
 and let $B$ denote the set of those chosen vertices
 (hence $B$ includes exactly one vertex in $A_u$ for each $u \in M'_v$).

   \begin{figure}[t]\centering
    \subfloat[$l_{vu}\geq l_{vi}/2$]{
     \includegraphics[scale=.6]{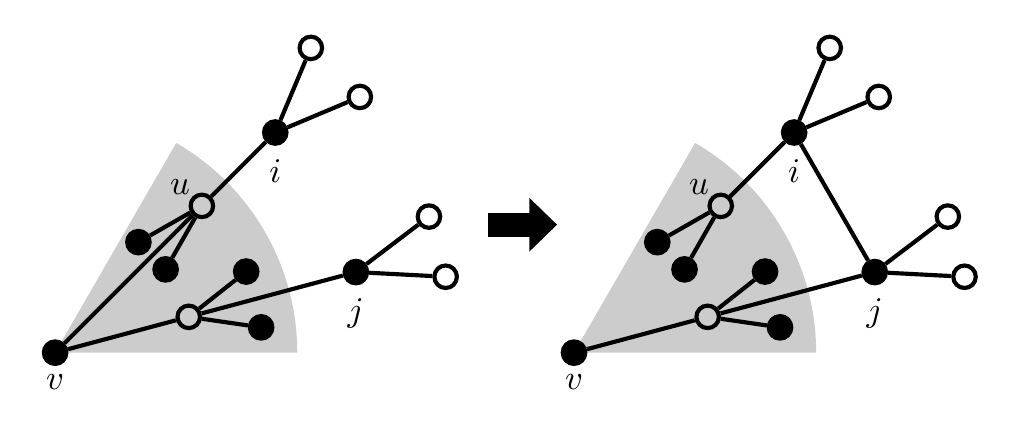}
    }
    \subfloat[$l_{ui}\geq l_{vi}/2$]{
     \includegraphics[scale=.6]{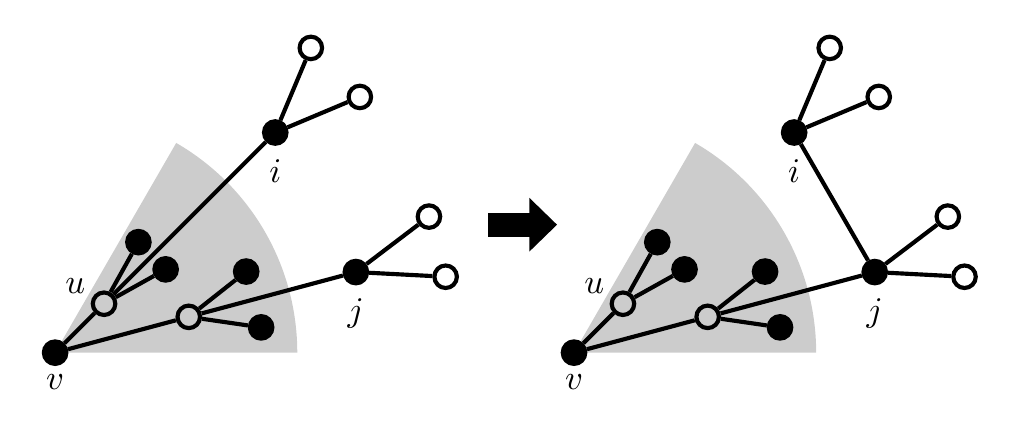}
    }
    \caption{Transformation of a tree $F$ when $l_{ij}<\max \{l_{vi},l_{vj}\}/2$}
   \label{fig.proof2}
 \end{figure} 

 Suppose there exist two vertices $i,j \in B$ such that $l_{ij} < \max\{l_{vi},l_{vj}\}/2$.
 Let $l_{vi}=\max\{l_{vi},l_{vj}\}$.
 Let $u$ denote the common neighbor of $v$ and $i$.
 Then, $l_{vu}$ or $l_{ui}$ is at least $l_{vi}/2$.
 If $l_{vu} \geq l_{vi}/2$, then replace edge $vu$ by $ij$ in $F$
 (see Figure~\ref{fig.proof2}(a)).
 Otherwise, replace edge $ui$ by $ij$ in $F$
 (see Figure~\ref{fig.proof2}(b)).
 Let $F'$ denote the tree obtained by this replacement.
 Then $(F',U)$ is another optimal solution,
 all inner nodes of $(F',U)$ are also inner nodes of $(F,U)$
 ($u$ is an inner node of $(F,U)$, but it may not be an inner node of $(F',U)$),
 and 
 $\sum_{e\in F'}l_{e} < \sum_{e \in F}l_e$
 holds.
 Since this contradicts condition (a) or (b), there exists no such pair
 of vertices $i,j\in B$.

 We divide $B$ into $B':=\{i \in B \mid l_{vi} \leq 1.41\}$
 and  $B'':=\{i \in B \mid  1.41 < l_{vi}\}$.
 Notice that $1/1.41 \leq l_{vi}/l_{vj}\leq 1.41$
 holds for any $i,j \in B'$.
 Hence, by Lemma~\ref{lem.geometry},
 $|B'|\leq \lfloor 2\pi/\arccos(1/2.82+4.23/8) \rfloor =12$.
 Moreover, $1.41/2 \leq l_{vi}/l_{vj}\leq 2/1.41$
 holds for any $i,j \in B''$.
 Hence, by Lemma~\ref{lem.geometry},
 $|B''|\leq \lfloor 2\pi/\arccos(1.41/4 + 3/5.64) \rfloor =12$.
 Since $|M_v'| \leq |B|=|B'| + |B''|\leq 24$, this proves the lemma.
\end{proof}

In the remainder of this subsection, we assume that $G$ is not necessarily
a unit disk graph, but there is an optimal solution $(F,U)$ for the \vcst
such that the degree of each vertex $v \in U$ is at most a constant
$\alpha$
in the tree $F - (L(F)\setminus U)$.
Based on this assumption,
we reduce the \vcst to another optimization problem.
First, let us define the problem used in the reduction.

  \begin{definition}[Connected dominating set problem]
   Let $G=(V,E)$ be an undirected graph, and let $T \subseteq V$ be a set of
  terminals.
  Each edge $e$ is associated with the length $l(e) \in \Rset_+$,
   each vertex $v$ is associated with the weight
   $w(v) \in \Rset_+$,
   and $l(e) \leq \min\{w(u),w(v)\}$ holds for each edge $e=uv \in E$.
  The problem seeks a pair of a tree $F \subseteq E$ and a vertex
   set $S \subseteq V$ such that
   $S$ dominates $T$ and $F$ spans $S$.
   Let $l(F)$ denote
   $\sum_{e \in F}l(e)$.
   The objective is to minimize $w(S) + l(F)$.
  \end{definition}

  We note that there are several previous studies of the connected dominating set problem~\cite{GuhaK98,ChengHLWD03,AmbuhlEMN06,ZouWXLDWW11}. However,
  the algorithms in those studies do not apply to our setting
  because they consider only the case $T=V$.

   \begin{theorem}
    \label{thm.reduction-unit}
    Suppose that there exists an optimal solution $(F,U)$ for the VC-weighted Steiner tree
    problem such that each node in $U$
    has a degree at most $\alpha$ on the tree $F$.
    If there is a $\beta$-approximation algorithm for the connected
    dominating set problem in a graph $G$, then there is an $(\alpha+1)\beta$-approximation
   algorithm for the \vcst with input graph $G$.
   \end{theorem}
  \begin{proof}
   Suppose that an instance $I$ of the \vcst
   consists of an undirected graph $G=(V,E)$,
   a terminal set $T \subseteq V$, and vertex weights $w\in \Rset_+^V$.
   We define the edge length $l(e)$ as $\min\{w(u),w(v)\}$ for each $e
   =uv \in E$,
   and define an instance $I'$ of the connected dominating set 
   problem
   from $G$, $T$, $w$, and $l$.
   We show that the optimal objective value of $I'$ is at most
   $\alpha+1$
   times that of $I$,
   and a feasible solution for $I$ can be constructed from the one for
   $I'$
   without increasing the objective value.
   Combined with the $\beta$-approximation algorithm for $I'$,
   these claims give an $(\alpha+1)\beta$-approximation
   algorithm for $I$.

   First, we prove that the optimal objective value of $I'$ is at most
   $\alpha+1$ times that of $I$.
   Let $(F,U)$ be an optimal solution for $I$.
   Then, the optimal objective value of $I$ is $w(U)$.
   Since $F$ spans $T$ and $U$ is a vertex cover of $F$,
   $U$ dominates $T$.
   Define $F':= F - (L(F)\setminus U)$.
   Since $F'$ is a tree spanning $U$, $(F',U)$ is a feasible solution for
   $I'$.
   If $e=uv \in F'$, then $u$ or $v$ is included in $U$,
   and $l(e)$ is at most $w(u)$ and $w(v)$.
   Hence $l(F') \leq \sum_{v \in U}w(v) d_{F'}(v)$.
   By assumption,
   $d_{F'}(v)\leq \alpha$ holds for each $v \in U$.
   Hence $l(F') \leq \alpha w(U)$.
   Since the objective value of $(F',U)$ in $I'$ is $l(F')+w(U)$,
   the optimal objective value of $I'$ is at most $(\alpha+1)w(U)$.

   Next, we prove that a feasible solution $(F,S)$ for $I'$ provides a
   feasible solution for $I$, and its
   objective value is at most that of $(F,S)$.
      Since $S$ dominates $T$, if a terminal $t \in T$ is not spanned by
   $F$,
   there is a vertex $v \in S$ with $tv \in E$.
   We let $F'$ be the set of such edges $tv$.
   Notice that $F \cup F'$ is a Steiner tree of the terminal set $T$.
   For each edge $e \in F$, choose an end vertex $v$ of $e$
   such that $l(e)=w(v)$.
   Let $S'$ denote this set of chosen vertices.
   Then, $S' \cup S$ is a vertex cover of $F \cup F'$.
   Hence $(F \cup F',S'\cup S)$ is feasible for $I$.
   Since $w(S'\cup S)\leq w(S)+l(F)$, the objective value of
   $(F \cup F',S'\cup S)$ is at most that of $(F,S)$.
  \end{proof}

\subsection{Algorithm for the connected dominating set problem with $T=V$}
\label{sec.spanning}

In the remainder of this section,
we present algorithms for the connected dominating set problem.
As a warm-up, we will first discuss the case $T=V$, which arises in the
reduction from the \vcsp.
We show that the problem admits a simple constant-factor approximation
algorithm for any
graphs in which the minimum-weight dominating set problem admits a constant-factor
approximation.
This class includes unit disk graphs~\cite{ZouWXLDWW11,FonsecaSF14}.
Below, we let $\beta$ denote the approximation factor for the minimum-weight
dominating set problem.

 Our algorithm first computes a $\beta$-approximate solution $S$ of the minimum
  dominating set of the graph.
  Then, it computes a $\beta'$-approximation of the minimum edge-weighted Steiner tree that spans $S$.
    Let $F$ denote the computed Steiner tree.
   Then, $(F,S)$ is our approximate solution for the connected
   dominating set problem.
   
    \begin{theorem}
     \label{thm.spanning}
     Suppose that
      the minimum-weight dominating set problem
     admits a $\beta$-approximation algorithm
     and the minimum edge-weighted Steiner tree problem
     admits a $\beta'$-approximation algorithm
     for a graph $G$.
    The there exists a $\beta (\beta'+1)$-approximation algorithm for the
    connected dominating set problem with the graph $G$.
    \end{theorem}
  \begin{proof}
   Let $(F,S)$ denote a solution output by the algorithm,
   and let $(F^*,S^*)$ be an optimal solution for the problem.
   Since $S^*$ is a dominating set of $G$, $w(S)\leq \beta w(S^*)$ holds
   by the definition of $S$.
   Also, $F^*$ spans $S^*$. Hence, if a vertex in $S$ is included in $S^*$,
   then this vertex is spanned by $F^*$.

   Let $v \in S \setminus S^*$. Since $S^*$ is a dominating set of $G$,
   it includes a vertex  adjacent to $v$. Call such a vertex $v'$.
   If $v$ has more than one neighbor in $S^*$, choose one of them
   arbitrarily and call it $v'$.
   Let $F' = \{v v' \mid v \in S\setminus S^*\}$.
   Notice that all edges in $F'$ are incident to some vertex in $S$,
   and the degree of each vertex in $S$ is at most one in $F'$.
   Hence $l(F')\leq w(S)$.

   $F^* \cup F'$ is a connected subgraph that spans $S$.
   Therefore, $l(F) \leq \beta' l(F^* \cup F')\leq
   \beta' (l(F^*)+w(S))$.
   The objective value of the solution $(F,S)$
   is $l(F)+w(S) \leq \beta' l(F^*)+ (\beta'+1)w(S)\leq   \beta' l(F^*)+
   (\beta'+1)\beta w(S^*) \leq \beta(\beta'+1) (l(F^*)+w(S^*))$.
   Therefore, the approximation factor of $(F,S)$ is at most $\beta (\beta'+1)$.
  \end{proof}

  As mentioned above, in a unit disk graph, $\beta$ is a constant.
  The edge-weighted Steiner tree problem admits a constant-factor
  approximation algorithm for any graphs~\cite{GoemansORZ12,ByrkaGRS13}.
  Hence Theorem~\ref{thm.spanning} provides the following corollary.
  
  \begin{corollary}
   The \vcsp admits a constant-factor approximation algorithm in unit disk graphs.
  \end{corollary}
  
  \subsection{Algorithm for the connected dominating set problem}
  \label{sec.lp-rounding}
  
  We now provide a constant-factor approximation algorithm for the
  general case of 
  the connected dominating set problem in unit disk graphs.
  Our algorithm is based on the idea given by  Huang, Li, and Shi~\cite{HuangLS15}.

 \begin{figure}
  \centering
  \includegraphics[scale=.8]{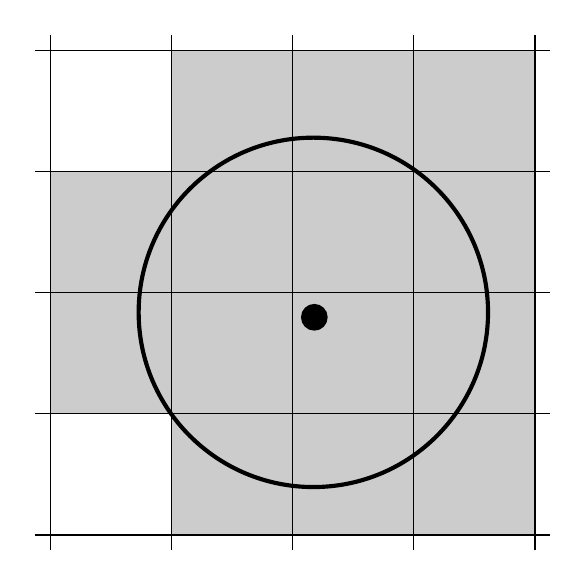}
  \caption{A unit disk and squares of sides $\sqrt{2}/2$}
  \label{fig.square}
 \end{figure}
  
   We say that a graph $G=(V,E)$ has a property \emph{$\pi$}
   if there is a partition $\Pi:=\{V_1,\ldots,V_k\}$ of the
   vertex set $V$ such that each $V_i \in \Pi$ induces a clique
   in $G$ and each vertex $v \in V$ satisfies
   $|\{i =1,\ldots,k \colon V_i \cap N_G[v]\neq
   \emptyset\}| \leq \theta$
   with some constant $\theta$.
   The following lemma proves that 
   a unit disk graph possesses this property.

   \begin{lemma}
    A unit disk graph has the property $\pi$ with a constant $\theta=14$.
   \end{lemma}
   \begin{proof}
    Divide the Euclidean plane into squares of side $\sqrt{2}/2$.
    We define each class $V_i$ of the partition $\mathcal{P}$
    as a set of vertices whose positions
    are on the same square
    in the geometric representation of $G$.
    If a vertex $u$ is on a side with more than one square, then we assign $u$
    to the upper-right square.
  Then, since any two vertices in the same square are within a unit distance, each class $V_i$ of $\mathcal{P}$ induces a clique in $G$.
  Moreover,
  the neighbors of a vertex
  belong to at most 14 classes of $\mathcal{P}$.
  This is because any unit disk intersects at most 14 squares; 
  see Figure~\ref{fig.square}.
  In the example shown in Figure~\ref{fig.square},
  the unit disk intersects the gray squares.
   The disk touches the square in the lower left,
    but we do not say that they intersect because
    a vertex on a border belongs
    to the upper-right square.
    Hence the unit disk graph satisfies property $\pi$ with
  $\theta=14$.
   \end{proof}  

   In the remainder,
   we 
   assume that 
   a vertex $r$ is spanned by the tree in an optimal solution to the problem.
   Although we do not know which vertex in $G$ is spanned, 
   we can guess it by applying the algorithm with setting each vertex in $V$ to $r$.

  For each $v \in V$, let $\mathcal{P}_v$ be the set of paths between $r$ and $v$.
  Under the assumption that $r$ is spanned by an optimal solution,
  the problem is relaxed to the following LP:
   \begin{align}\label{eq.lp}
    \begin{aligned}
     \parbox{6em}{minimize} & \sum_{v \in V}w(v)x(v) + \sum_{e \in E} l(e)y(e) \\
    \parbox{6em}{subject to} &
    \sum_{v \in N_G[t]} x(v) \geq 1 &&  \forall t \in T,\\
    & \sum_{P \in \mathcal{P}_v}f(P) = x(v) &&  \forall v \in V,\\
    & \sum_{P \in \mathcal{P}_v: e \in P}f(P) \leq y(e) &&  \forall e \in E, \forall v \in V,\\
    & x(v) \geq 0  && \forall v \in V,\\
   & y(e) \geq 0  && \forall e \in E,\\
   & f(P) \geq 0  && \forall P \in \bigcup_{v \in V}\mathcal{P}_v.
    \end{aligned}   
   \end{align}  
  Indeed, if $x \in \{0,1\}^V$ and $y \in \{0,1\}^E$, then the feasible
  solution $(x,y,f)$
  for \eqref{eq.lp} corresponds to a feasible solution to the connected
  dominating set problem.
  Here, $x(v)$ indicates if vertex $v$ is included in
  a dominating set $S$ (if $x(v)=1$, $v$ is included in $S$), and $y(e)$ indicates if
  edge $e$ is included in a tree $F$ that spans the dominating
  set $S$ (if $y(e)=1$, $e$ is included in $F$).
  Also, $f(P)$ represents the flow value along path $P$.
  The second constraint demands that the flow value between $r$ and $v$
  is at least $x(v)$,
  and the third constraint means that the flow between $r$ and $v$
  obeys the edge capacities $y$.
  Hence, if $x(v)=1$,
  one unit of flow runs between $r$ and $v$.
  This means that the minimum cut separating $v$ from $r$ with respect
  to the edge capacities $y$ has a capacity of at least 1.
  Hence the edge set $\{e \in E \colon y(e)=1\}$ connects $r$ and each
  vertex $v$ with $x(v)=1$.

  Although there are an exponential number of variables 
  in the LP~\eqref{eq.lp}, it can be converted into an equivalent
  formulation of polynomial size. Hence 
  an optimal solution $(x^*,y^*,f^*)$ for \eqref{eq.lp} can be computed
  in polynomial time.
  Our algorithm computes this, and then
  from this optimal fractional solution, it constructs a dominating set $S$
  and a tree $F$, as follows.

  We define $I$ as $\{i=1,\ldots,k \colon
  \sum_{v \in  V_i}x^*(v)\geq 1/\theta\}$.
  We restrict a dominating set to be included in $V_{I}:=\bigcup_{i \in I}V_i$.
  Namely, the dominating set $S$ computed by our algorithm
  is feasible to the following integer program:
   \begin{align}\label{eq.dominating}
    \begin{aligned}
     \parbox{6em}{minimize} & \sum_{v \in V}w(v)x(v)\\
    \parbox{6em}{subject to} &
     \sum_{v \in N[t] \cap V_I} x(v) \geq 1 &&  \forall t \in T,\\
     & x(v) \in \{0,1\}  && \forall v \in V.
    \end{aligned}   
   \end{align}  

   By replacing constraint $x(v) \in \{0,1\}$ with $x(v)\geq 0$,
   we obtain an LP relaxation of \eqref{eq.dominating}.
   By the following lemma, the optimal objective value of this relaxation can be bounded by 
   $\theta$ times the weight of $x^*$.

   \begin{lemma}
    The LP relaxation of \eqref{eq.dominating}
    has the optimal objective value that is at most $\theta \sum_{v \in V}w(v)x^*(v)$.
   \end{lemma}
   \begin{proof}
    For each terminal $t \in T$, there is a class $V_i$ of $\mathcal{P}$ with
    $\sum_{v \in N[t] \cap V_i} x^*(v) \geq 1/\theta$ (and hence $i \in I$),
  because $\sum_{v \in N[t]}x^*(v) \geq 1$ holds and 
  the vertices in $N[t]$ belong to at most $\theta$ classes of
  partition $\mathcal{P}$.
  This implies that $\theta x^*$ is feasible to the LP relaxation of
   \eqref{eq.dominating}. Therefore, the optimal objective value of the relaxation
   is at most $\theta \sum_{v \in V}w(v)x^*(v)$.
   \end{proof}

   Problem \eqref{eq.dominating} is a special case of the geometric
   set cover problem, in which
   the ground set is a set of points on a Euclidean plane,
   and each set is represented by
   a unit disk.
   Several constant-factor approximation algorithms are known for this
   problem, including a PTAS due to Li and Jin~\cite{LiJ15}.
   However, most of them are not useful for our purpose because
   we require bounding the weight of the output solution
   with regard to the optimal value of the LP relaxation of \eqref{eq.dominating}.
   As far as we know, the only constant-factor approximation algorithm
   satisfying this requirement is 
   due to Chan~et~al.~\cite{ChanGKS12}.
   Let $\gamma$ be the approximation factor of this algorithm.
   Then, this algorithm computes $S \subseteq V_I$
   such that $T$ is dominated by $S$ and $w(S) \leq \gamma\theta \sum_{v \in V}w(v)x^*(v)$.
   In our algorithm for the connected dominating set problem,
   the dominating set is defined 
   as the vertex set $S$ computed by the algorithm for the geometric set
   cover problem.

   Our algorithm then computes a tree that spans $r$ and $S$.
   Let us explain how to compute the tree.
   For each $i \in I$,
   we choose an arbitrary vertex in $V_i$ and call it $v_i$.
   We use an algorithm for the
   Steiner tree problem to construct a minimum-length tree that spans $r$ and all vertices
   $v_i$, $i \in I$. An LP relaxation of this Steiner tree problem
   can be written as follows:
   \begin{align}\label{eq.steinertree}
    \begin{aligned}
     \parbox{6em}{minimize} & \sum_{e \in E}l(e)y(e)\\
    \parbox{6em}{subject to} &
     \sum_{P \in \mathcal{P}_{v_i}}f(P) = 1 && \forall i \in I,\\
     & \sum_{P \in \mathcal{P}_{v_i}: e\in P}f(P) \leq y(e) && \forall  e\in E,  \forall i \in I,\\
     & y(e) \geq 0  && \forall e \in E,\\
     & f(P) \geq 0 && \forall P \in \bigcup_{i \in I} \mathcal{P}_{v_i}.
    \end{aligned}   
   \end{align}  

   We note that $(y^*,f^*)$ is not necessarily feasible to
   \eqref{eq.steinertree}.
   Nevertheless, we can
   bound the optimal objective value of \eqref{eq.steinertree}.

   \begin{lemma}
    The optimal objective value of \eqref{eq.steinertree} is at most $\theta (\sum_{v \in V}w(v)x^*(v)+\sum_{e \in E}l(e)y^*(e))$.
   \end{lemma}
   \begin{proof}
    We define a feasible solution for \eqref{eq.steinertree}
    from $(y^*,f^*)$.
    First, initialize $(y,f)$ to $(y^*,f^*)$.
   Let $i \in I$.
   Then, $\sum_{v \in V_i}\sum_{P \in \mathcal{P}_v}f^*(P) \geq \sum_{v \in
    V_i}x^*(v) \geq 1/\theta$.
    Recall that there exists an edge $v v_i$ for each vertex $v \in V_i \setminus v_i$.
    If $\sum_{P \in \mathcal{P}_v}f^*(P)=\epsilon$ for $v \in V_i \setminus v_i$,
   we increase $y(vv_i)$ by $\epsilon$, increase $f(P \cup \{vv_i\})$ by $f(P)$,
   and set $f(P)=0$, for every $P\in \mathcal{P}_v$.
   Notice that $l(vv_i)$ does not exceed $w(v)$.
   Hence the increase of $y(vv_i)$ costs $l(vv_i)\epsilon \leq w(v)\epsilon =w(v)x^*(v)$.
   We do this for every $i \in I$ and for every vertex $v \in V_i \setminus v_i$.
   At the termination of this procedure,
   $\sum_{e\in E}l(e)y(e) \leq \sum_{v \in V}w(v)x^*(v) + \sum_{e\in E}l(e)y^*(e)$,
   and 
   $\sum_{P \in \mathcal{P}_{v_i}}f(P) = \sum_{v \in V_i}\sum_{P \in
   \mathcal{P}_{v}}f^*(P) \geq 1/\theta$.
   We define $(y',f')$ as
   $(\theta y, \theta f)$. Then, $(y',f')$ is feasible for \eqref{eq.steinertree},
   and its objective value does not exceed $\theta (\sum_{v \in
    V}w(v)x^*(v)+\sum_{e \in E}l(e)y^*(e))$,
    completing the proof.
   \end{proof}

   Goemans and Bertsimas~\cite{GoemansB93}
   showed that a Steiner tree of length at most twice
   the optimal objective value of \eqref{eq.steinertree}
   can be computed from a minimum spanning tree in the metric completion
   on the terminal set.
   Namely, there is an algorithm that computes a tree $F$ spanning $r$
   and all vertices $v_i$, $i \in I$,
   such that $l(F) \leq 2 \theta (\sum_{v \in V}w(v)x^*(v)+\sum_{e \in  E}l(e)y^*(e))$.
   $F$ may not span a vertex $v \in S$.
   For such a vertex $v$, we add an edge joining $v$ with $v_i$, where
   $i$ is an index such that $v \in V_i$.
   Notice that $G$ contains an edge $vv_i$,
   because $V_i$ induces a clique.
   Let $F'$ denote the set of these added edges.
   Notice that $l(F')\leq w(S)$.
   Our algorithm outputs $(F\cup F',S)$ as a solution for the connected
   dominating set problem.
   Recall that we are assuming here that $r$ is spanned by an optimal solution.
   When we implement the algorithm, we apply it to the vertices
   in $V$ as $r$, and define the output as the best of the obtained solutions.
   
    \begin{theorem}
     \label{thm.connfl}
    The solution $(F\cup F',S)$ computed by the above algorithm is a
    $2(\gamma+1) \theta$-approximate solution for the connected
     dominating set problem.
    \end{theorem}
   \begin{proof}
    $S$ dominates $T$. Moreover,
    $F$ connects each vertex $v_i$, $i \in  I$, to $r$, and $F'$ connects
    each vertex $v \in S\cap V_i$ to $v_i$.
    Hence $(F\cup F',S)$ is feasible for the connected dominating set
    problem.

    As noted above, we have
    $l(F')\leq w(S) \leq \gamma \theta \sum_{v \in V}w(v)x^*(v)$, and 
    $l(F) \leq 2 \theta (\sum_{v \in V}w(v)x^*(v)+\sum_{e \in
    E}l(e)y^*(e))$.
    Hence the objective value $l(F)+l(F')+w(S)$ is at most $2(\gamma+1)\theta$
    times the optimal objective value of \eqref{eq.lp}.
    Since \eqref{eq.lp} relaxes the connected dominating set problem,
    this proves the theorem.
   \end{proof}

   Recall that
   $\theta=14$ and 
   $\gamma$
   is the approximation factor of the geometric set cover algorithm of
   Chan~et.~al.~\cite{ChanGKS12}.
   It is shown in \cite{ChanGKS12} that $\gamma$ is a
   constant, although the bound on $\gamma$ is not stated explicitly.
   Theorem~\ref{thm.connfl} has the following corollary.

    \begin{corollary}
     \label{cor.unitdisk}
    The \vcst admits a constant-factor approximation algorithm in unit disk graphs.
    \end{corollary}

  \section{Steiner tree activation problem in graphs excluding a fixed minor}
  \label{sec.planer}

  In this section, we present a constant-factor approximation algorithm
  for the Steiner tree activation problem in graphs excluding a fixed
  minor.
  In particular, our algorithm is a 11-approximation for planar graphs.
  
  Our algorithm is based on the reduction 
  mentioned in
  Theorem~\ref{thm.reduction}.
  We reduce the problem to the vertex-weighted Steiner tree problem by
  using that reduction,
  and we solve the obtained instance by using the constant-factor approximation algorithm
  proposed by Demaine, Hajiaghayi, and Klein~\cite{DemaineHK09a}
  for the vertex-weighted Steiner tree problem in 
  graphs excluding a fixed minor.
  We prove that this achieves a constant-factor approximation  for the
   Steiner tree activation problem when the input graph is
   $H$-minor-free for some graph $H$ such that $|V(H)|$ is a constant.

   This seems to be an easy corollary to
  Demaine et al.,
  but it is not so
  because the reduction
  does not preserve the $H$-minor-freeness of the input graph.
 Let $G$ be the
  graph obtained by removing one edge from $K_5$.
  It is easy to check that $G$ is planar.
  We consider the \vcsp over $G$.
  The reduction transforms $G$ into another graph $G'$ on the vertex set
  $V(G) \cup \{v_{\circ},v_{\bullet}\colon v \in V(G)\}$.
  Refer to the proof of Theorem~\ref{thm.reduction} for 
  the definition of the edge set of $G'$.
  Notice that the subgraph of $G'$ induced by
  $V_{\bullet}:=\{v_{\bullet}\colon v \in V(G)\}$ is isomorphic to $G$.
  Let $u$ be an arbitrary vertex in $V(G)$
  that is not an end vertex of the
  removed edge.
  The subgraph of $G'$ induced by $V_{\bullet} \cup u_{\circ}$ contains
  a subgraph isomorphic to a subdivision of $K_5$,
  and hence $G'$ is not planar.

  As indicated by this example, the reduction does not preserve
  the $H$-minor-freeness.
  In spite of this, we can prove that the approximation guarantee
  given by Demaine et al.\ 
  extends to the graphs constructed
  from a $H$-minor-free graph by the reduction.

  We recall that
  the reduction constructs a graph $G'$  on the vertex set
  $T \cup \{v_i \colon v \in V, i \in W\}$ from the input graph
   $G=(V,E)$ and the monotone activation functions $f_{uv}\colon W \times W
   \rightarrow \{\top,\bot\}$, $uv \in E$.
  We denote the vertex set $\{v_i \colon i \in W\}$ defined from an
  original vertex $v \in V$ by $U_v$.
  Let $U$ denote $\bigcup_{v \in V}U_v$.

  First, let us illustrate how the algorithm of Demaine~et~al.\ behaves
  for $G'$.
  The algorithm maintains a vertex set $X \subseteq T \cup U$,
  where $X$ is initialized to $T$ at the beginning.
  Let $\Afam(X) \subseteq 2^X$ denote the family of connected components
  that include some terminals
  in the subgraph of $G'[X]$.
  We call each member of $\Afam(X)$ an \emph{active set}.
  The algorithm consists of two phases, called the increase phase and the reverse-deletion
  phase.
  In the increase phase,
  the algorithm iteratively adds vertices to $X$ until $|\Afam(X)|$
  is equal to one. This implies that, when the increase phase terminates, the subgraph induced by $X$
  connects all of the terminals.
  In the reverse-deletion phase,
  $X$ is transformed into an inclusion-wise minimal vertex set that
  induces a Steiner tree.
  This is done by repeatedly removing vertices from $X$ in the reverse
  of the order in which they were added.

  Let $\bar{X}$ be the vertex set $X$ when the algorithm terminates,
  and let $X$ be the vertex set at some point during the increase phase.
  We denote $\bar{X}\setminus X$ by $\bar{X}'$.
  Note that $\bar{X}'$ is a minimal augmentation of $X$ such that
  $X \cup \bar{X'}$ induces a Steiner tree.
  Each $Y \in \Afam(X)$ is disjoint from $\bar{X}'$, because
  $Y \subseteq X$.
  Demaine~et~al. showed the following analysis of their algorithm.

   \begin{theorem}[\cite{DemaineHK09a}]\label{thm:demaine}
    Let $X$ be a vertex set maintained at some moment in the increase
    phase,
    and let $\bar{X}'$ be a minimal augmentation of $X$ so that $X \cup
    \bar{X}'$
    induces a Steiner tree.
    If there is a number $\gamma$ such that
    $\sum_{Y \in \Afam(X)} |\bar{X}' \cap N(Y)| \leq \gamma |\Afam(X)|$
    holds for any $X$ and $\bar{X}'$,
    the algorithm of Demaine~et~al. achieves an approximation factor
    $\gamma$.
   \end{theorem}

   In $G'[\bar{X}' \cup (\bigcup_{Y \in \Afam(X)}Y)]$,
   contract each $Y \in \Afam(X)$ into a single vertex, 
   discard all edges induced by $\bar{X}'$ and
   all isolated vertices in $\bar{X}'$,
   and replace multiple edges by single edges.
   This gives us a simple bipartite graph with the bipartition
   $\{A,B\}$ of the vertex set,
   where each vertex in $A$ corresponds to an active set,
   and $B$ is a subset of $\bar{X}'$.
   Let $D$ denote this graph.
   This construction of $D$ is illustrated in
   Figure~\ref{fig.bipartitegraph}.
   We note that $\sum_{Y \in \Afam(X)} |\bar{X}' \cap N(Y)|$ is equal to the number
   of edges in $D$.
   Hence, by Theorem~\ref{thm:demaine},
   if the number of edges is at most a constant factor of $|A|$,
   the algorithm achieves a constant-factor approximation.

   Demaine~et~al.\ proved that
   $|B| \leq 2|A|$, and $D$ is $H$-minor-free if $G$ is $H$-minor-free.
   By \cite{Kostochka84,Thomason01}, these two facts imply that
   the number of edges in $D$ is $O(|A| |V(H)| \sqrt{\log |V(H)|})$.
   When $G$ is planar, 
   together with Euler's formula and the fact that $D$ is bipartite,
   they imply that the number
   of edges in $D$ is at most $6|A|$.

   The proof of Demaine~et~al.\ for $|B| \leq 2|A|$ can be carried to
   our case. However, $D$ is not necessarily $H$-minor-free even if $G$
   is $H$-minor-free.
   Nevertheless, we can bound the number of edges in $D$, as follows.

    \begin{lemma}
     \label{lem.D}
     Suppose that the given activation function is monotone.
     If $G$ is $H$-minor-free, the number of edges in $D$ is
     $O(|A||V(H)|\sqrt{\log |V(H)|})$.
     If $G$ is planar, the number of edges in $D$ is at most $11 |A|$.
    \end{lemma}
   
   The following theorem is immediate from Theorem~\ref{thm:demaine} and Lemma~\ref{lem.D}.
  
   \begin{theorem}
    \label{thm.planer}
    If an input graph is $H$-minor-free for some graph $H$,
    then the Steiner tree activation problem with a monotone activation function
    admits an $O(|V(H)| \sqrt{\log |V(H)|})$-approximation algorithm.
    In particular, if the input graph is planar,
    then the problem 
    admits a $11$-approximation algorithm.
   \end{theorem}

  In the rest of this section, we prove Lemma~\ref{lem.D}.
  We first provide several preparatory lemmas.

  \begin{lemma}
   \label{lem.reduction-prop}
   If $G'$ includes an edge $u_{i}v_j$ for some $u,v\in V$ and $i,j\in
   W$,
   then $G'$ also includes an edge
    $u_{i'}v_{j'}$ for any $i',j' \in W$ with $i' \geq i$ and $j' \geq j$.
  \end{lemma}
  \begin{proof}
   The lemma is immediate from the construction of $G'$ and the
   assumption that each edge in $G$ is associated with a monotone activation function.
  \end{proof}
  
  \begin{lemma}\label{lem.planer1}
    $\bar{X}$ does not contain any two distinct copies 
   of an original vertex.
  \end{lemma}
  \begin{proof}
   For the sake of a contradiction, suppose
  that $v_i,v_j \in \bar{X}$ for some $v \in V$ and $i,j \in W$ with
  $i < j$.
  If an edge $u_kv_i$ exists in $G'$,
   then another edge $u_k v_j$ also exists by Lemma~\ref{lem.reduction-prop}.
  This means that $\bar{X}\setminus v_i$
  induces a Steiner tree in $G'$, which contradicts the minimality of $\bar{X}$.
 \end{proof}

  \begin{lemma}\label{lem.planer2}
   Let $Y,Y' \in \Afam(X)$ with $Y \neq Y'$.
   If $Y \cap U_v \neq \emptyset$ for some $v \in V$, then $Y' \cap U_v=\emptyset$.
  \end{lemma}
  \begin{proof}
   Suppose that
   $Y \cap U_v \neq \emptyset \neq Y' \cap U_v$.
   Let $v_i \in Y$ and $v_j \in Y'$ with $i< j$.
   A vertex adjacent to $v_i$ is also adjacent to $v_j$ in $G'$ by Lemma~\ref{lem.reduction-prop}.
   By the definition, $Y$ induces a connected component of $G'[X]$ that includes a
   terminal $t$.
   Hence $v_i$ has at least one neighbor in $Y$.
   This implies that $v_i$ and $v_j$ are connected in $G'[X]$.
   This contradicts the fact that $Y$ and $Y'$ are different connected components
   of $G'[X]$.
  \end{proof}

    \begin{figure}
     \centering
    \includegraphics[scale=.85]{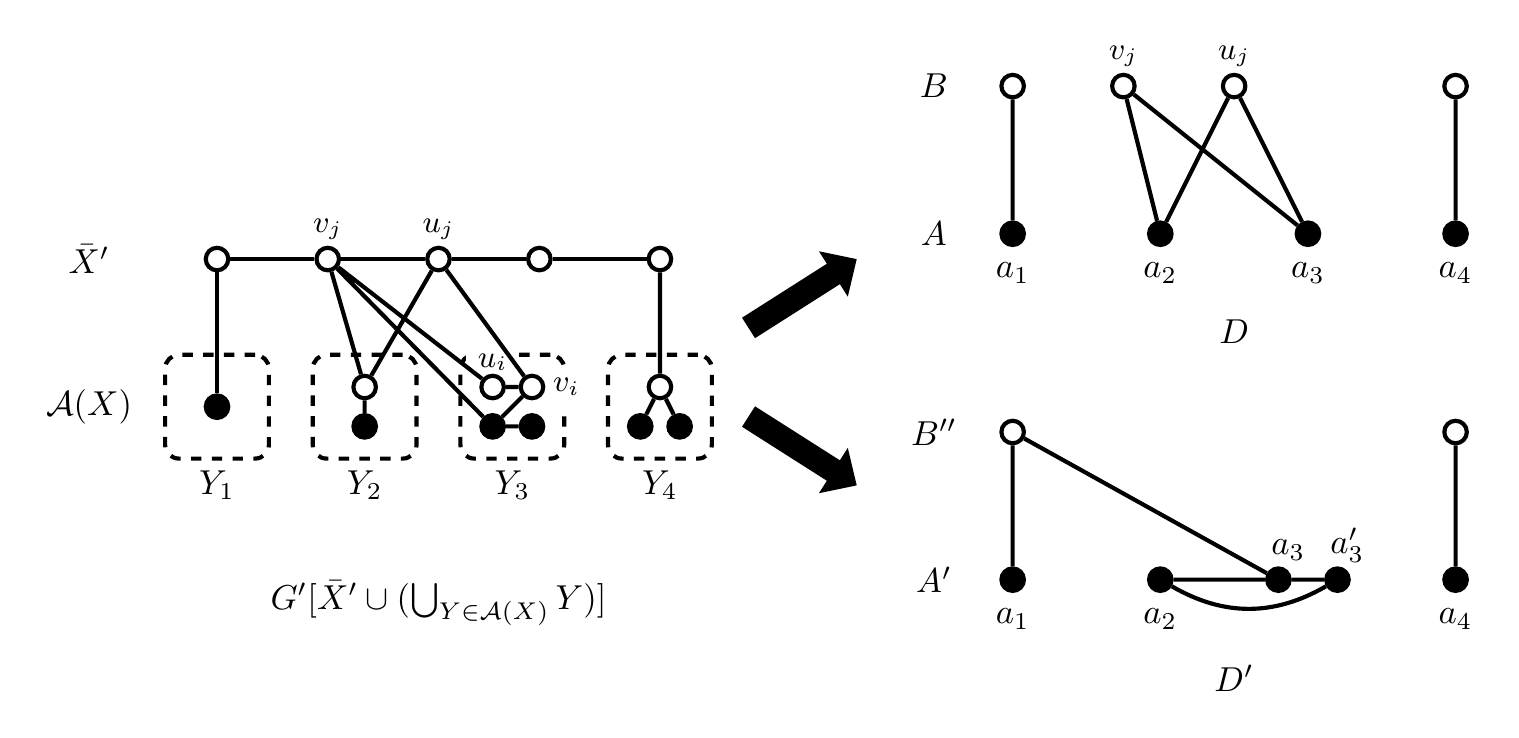}
    \caption{An example of $G'[\bar{X}' \cup (\bigcup_{Y \in
    \Afam(X)}Y)]$, $D$, and $D'$; in construction of $D'$, $\bar{Y}_3$
    is divided into two subsets, one of which contains $u_j$ and the
    other contains $v_j$; the former is shrunken into $a_3$ and the latter
    is shrunken into $a'_3$}
    \label{fig.bipartitegraph}
    \end{figure}

    To prove Lemma~\ref{lem.D}, we modify $D$ to obtain an
    $H$-minor-free graph $D'$. By this modification,
    the number of vertices does not increase,
    and the decrease of the number of edges is bounded.
    Thus, the $H$-minor-freeness of $D'$ implies that the number of
    edges of $D$ is bounded in terms of the number of its vertices.
    $D'$ is constructed by removing copies of the same vertices
    carefully. This is motivated by the observation that $D$ is not
    $H$-minor-free
    when it includes more than one copy of the same vertex
    of the original graph $G$. The construction is a bit
    complicated
    because copies of a vertex may appear both in an active set and in $\bar{X}'$.

    Now, we explain how to construct $D'$.
  Consider $Y \in \Afam(X)$ and $v \in V$ such that $Y \cap U_v \neq \emptyset$.
  Let $v_{i}$ be the vertex that has the largest subscript in $Y  \cap
  U_v$ (i.e., $i=\max\{i' \in W \colon v_{i'} \in Y \cap U_v\}$).
  Then, from $Y$, we remove all vertices in $Y\cap U_v$ but $v_i$.
  Moreover, if a copy $v_j$ of $v$ is included in $B$,
  we replace $v_i$ by $v_j$.
  Notice that
    $j > i$  holds in this case by Lemma~\ref{lem.reduction-prop},
    and $B$ does not include  more than one copy of $v$ because of Lemma~\ref{lem.planer1}.
  Let $\bar{Y}$ denote the vertex set obtained from $Y$ by doing these operations
  for each $v \in V$ with $Y \cap U_v \neq \emptyset$.
  $\bar{Y}$ induces a connected subgraph of $G'$
  because of Lemma~\ref{lem.reduction-prop}.
  
  We let $V_B$ denote $\{v \in V\colon B \cap U_v \neq \emptyset\}$,
  and let $V_{B,Y}$ denote
  $\{v \in V_B \colon \bar{Y} \cap U_v \neq \emptyset\}$ for each $Y
  \in \Afam(X)$.
  Moreover, let $B'$ denote $B \cap \{U_v \colon v\in  \bigcup_{Y \in
  \Afam(X)} V_{B, Y}\}$,
  and $B''$ denote $B \setminus B'$.
  In other words, each vertex $v_j \in B$  belongs to $B'$
  if and only if some copy $v_i$ of the same original vertex $v\in
  V$
  is contained by an active set in $\Afam(X)$.

  If $k:=|V_{B,Y}|\geq 2$,
  we divide $\bar{Y}$ into $k$ subsets
  such that
  the copies of the vertices in  $V_{B,Y}$ belong to different
  subsets,
  and each subset induces a connected subgraph of $G'$.
  Let $\Afam'(X)$ denote the family of vertex sets obtained by doing these
  operations to all active sets in $\Afam(X)$.
  Notice that $|\Afam'(X)|= |\Afam(X)| + \sum_{Y \in \Afam(X)}\max\{0,|V_{B,Y}|-1\}$.
  Lemma~\ref{lem.planer2} indicates that, if a vertex $v \in V_B$ belongs
  to $V_{B,Y}$ for some $Y \in \Afam(X)$,
  then it does not belong to $V_{B,Y'}$ for any $Y' \in
  \Afam(X)\setminus \{Y\}$.
  Thus, $\sum_{Y \in \Afam(X)}|V_{B,Y}| \leq |B'|$, and hence
  $|\Afam'(X)| \leq |\Afam(X)|+|B'|$.

  We shrink each $Z \in \Afam'(X)$ into a single vertex
  in the induced subgraph $G'[B'' \cup (\bigcup_{Z \in
  \Afam'(X)} Z)]$ of $G'$,
  and convert the obtained graph into a simple graph
  by removing all self-loops and by replacing multiple edges with single edges.
  Let $A'$ denote the set of vertices obtained by shrinking vertex sets
  in $\Afam'(X)$, and let
  $D'$ denote the obtained graph (with the vertex set $A' \cup B''$).
  See Figure~\ref{fig.bipartitegraph} for an illustration of this construction.

  We note that the division of $\bar{Y}$ into $k=|V_{B,Y}|$ subsets
  is required for showing that 
  the number of edges in $D'$ is not too smaller than that in $D$.
  If we shrink each $\bar{Y}$ instead of each $Z\in \Afam'(X)$,
  then two edges in $D$ (e.g., edges $a_2v_j$ and $a_2u_j$ in the
  example of Fig.~\ref{fig.bipartitegraph})
  may become parallel by the shrinking, and one of them is removed
  from $D'$. The removal of such edges may decrease the number of edges
  in $D'$ too much.

  We observe that $D'$ is $H$-minor-free in the following lemma.
  
   \begin{lemma}
    \label{lem.minorfree}
   If $G$ is $H$-minor-free, then $D'$ is $H$-minor-free.
   \end{lemma}
  \begin{proof}
  By Lemma~\ref{lem.planer1} and the construction of $\Afam'(X)$,
  each vertex in $V$ has at most one copy in $B'' \cup (\bigcup_{Z \in
   \Afam'(X)} Z)$.
   If $G'[B'' \cup (\bigcup_{Z \in \Afam'(X)} Z)]$ includes an edge
   $u_iv_j$ for $u_i \in U_u$ and $v_j \in U_v$, then $G$ also includes
   an edge $uv$.
   Thus $G'[B'' \cup (\bigcup_{Z \in \Afam'(X)} Z)]$ is isomorphic to a subgraph of $G$.
   Since each $Z\in \Afam'(X)$  induces a connected subgraph of $G'$,
   the graph $D'$ (constructed from $G'[B'' \cup (\bigcup_{Z \in
   \Afam'(X)} Z)]$ by shrinking each $Z\in \Afam'(X)$)
   is a minor of $G$.
   Hence if $G$ is $H$-minor-free, $D'$ is also $H$-minor-free.
  \end{proof}

  The following lemma gives a relationship between $D$ and $D'$.

    \begin{lemma}
     \label{lem.num_edges}
    If $l$ is the number of edges in $D'$,
    then $D$ contains at most $l+|B'|$
    edges.
    \end{lemma}
\begin{proof}
   Let $av_j$ be an edge in $D$
  that joins vertices $a \in A$ and $v_j \in B$.
 Suppose that $a$ is a vertex obtained by shrinking $Y \in \Afam(X)$,
 and $v_j$ is a copy of $v \in V$.
 Remember that $v_j$ belongs to either $B'$  or $B''$.
 If $v_j \in B'$,
 it is contained by a vertex set in $\Afam'(X)$, denoted by $Z_v$.
 We consider the following three cases:
 \begin{enumerate}
  \item $v_j \in B'$ and $Z_v \subseteq \bar{Y}$
  \item $v_j \in B'$ and $Z_v \not\subseteq \bar{Y}$
  \item $v_j \in B''$
 \end{enumerate}

 In the second case,
 an edge in $D'$ joins vertices obtained by shrinking $Z_v$ and a subset
 of $\bar{Y}$.
 In the third case, $v_j$ exists in $D'$, and $D'$ includes an edge that
 joins $v_j$ and 
 the vertex obtained by shrinking a subset of $\bar{Y}$.
 Thus $D'$ includes an edge corresponding to $av_j$ in these two cases.
 We can also observe that no edge in $D'$ corresponds to more than two
 such edges $av_j$. This is because $Z_u \neq Z_v$ for any distinct
 vertices $u_i$ and $v_j$ in $B'$ by the construction of $\Afam'(X)$.
 
 In the first case, $D'$ may not contain an edge corresponding to $av_j$.
 However, the number of such edges is at most $|B'|$ in total
 because $\bar{Y}$ are uniquely determined from $v_j$ in this case.
 Therefore, the number of edges in $D$ is at most $l+|B'|$.
\end{proof}

We now prove Lemma~\ref{lem.D}.

 \begin{proof}[Lemma~\ref{lem.D}]
  The number of vertices in $D'$
  is at most $|\Afam'(X)|+|B''| \leq |\Afam(X)|+|B'|+|B''|=|A|+|B|$.
  As we mentioned, we can prove $|B|\leq 2|A|$ similar to Demaine~et~al.~\cite{DemaineHK09a}.
  Hence $D'$ contains at most $3|A|$ vertices.
  By Lemma~\ref{lem.minorfree}, $D'$ is $H$-minor-free.
  It is known~\cite{Kostochka84,Thomason01} that 
  the number of edges in an $H$-minor-free graph with $n$ vertices
  is $O(n|V(H)| \sqrt{\log |V(H)|})$.
   Therefore, the number of edges in $D'$
  is $O(|A||V(H)| \sqrt{\log |V(H)|})$.
  By Lemma~\ref{lem.num_edges}, this implies that the number of edges in $D$
  is $|B'|+O(|A||V(H)| \sqrt{\log |V(H)|})=O(|A||V(H)| \sqrt{\log |V(H)|})$.
  This fact and Theorem~\ref{thm:demaine} prove the former part of the lemma.
  
  If $G$ is planar,
   by Euler's formula, the number of edges in $D'$
  is at most $3(|A|+|B|)$.
  Hence, by Lemma~\ref{lem.num_edges}, the number of edges in $D$
  is at most $3(|A|+|B|) + |B'| \leq 3|A|+4|B| \leq 11|A|$.
  The latter part of the lemma follows from this fact and  Theorem~\ref{thm:demaine}. 
 \end{proof}

   \section{Conclusion}
   \label{sec:conclusion}
   In this paper, we formulate
   the VC-weighted Steiner tree problem,
   a new variant of the vertex-weighted
   Steiner tree and the Steiner tree activation problems.
   We proved that it is NP-hard for unit disk graphs and planar graphs.
   We also presented constant-factor approximation
   algorithms
   for the \vcst in unit disk graphs and for the Steiner tree
   activation problem in graphs excluding a fixed minor.

   An interesting future work is to investigate VC-weighted spanning
   tree or \vcst with unit weights for unit disk graphs and planar
   graphs.
   We do not know whether these problems are NP-hard or admit exact
   polynomial-time algorithms.
   Finding a constant-factor approximation algorithm for the Steiner
   tree activation problem in unit disk graphs also remains an open
   problem.

   \section*{Acknowledgements}
    The first author was supported by JSPS KAKENHI Grant Number JP17K00040.

\bibliographystyle{abbrv}
\bibliography{cds}

\end{document}